\documentclass{article}
\usepackage{amssymb}
\usepackage{python}

\usepackage[english]{babel}
\usepackage[utf8x]{inputenc}
\usepackage[T1]{fontenc}
\usepackage{xcolor}  

\usepackage[title]{appendix}
\usepackage{framed}

\usepackage{color}
\usepackage{algorithm}
\usepackage{algpseudocode}
\usepackage{float}
\usepackage{enumerate}
\usepackage{babel}

\usepackage{ulem}

\usepackage{fontawesome}

\definecolor{lightgray}{rgb}{0.9, 0.9, 0.9}

\usepackage{amsmath}
\usepackage{amssymb}
\usepackage{amsthm}
\usepackage{proof}
\usepackage{graphicx}
\usepackage{subcaption}
\usepackage[colorinlistoftodos]{todonotes}

\newtheorem{dfn}{Definition}

\newtheorem{rem}{Remark}

\newtheorem{thm}{Theorem}

\begin{document}
\title{Quantum advantage for
combinatorial optimization problems, Simplified}   
\author{Mario Szegedy, Rutgers University, CS department} 

\maketitle

\begin{abstract}
We observe that fault-tolerant
quantum computers have an optimal advantage over classical computers in
approximating solutions to many NP optimization problems. This observation however
gives nothing in practice.
\end{abstract}

\section{Exposition}

Niklas Pirnay, Vincent Ulitzsch, Frederik Wilde,
Jens Eisert and Jean-Pierre Seifert~\cite{s2022} in their Theorems 3.5 and 4.3 prove, that:

\begin{thm}\label{Pirnayatal}
Assuming the hardness of inverting the RSA function, there exists 
a subset, called FC-RSA, of instances for the formula coloring problem (FC)
such that:
\begin{enumerate}
\item However large (but fixed) $\alpha\ge 1$ and $0\le  \beta < 1/2$ are, there is no classical probabilistic polynomial-time algorithm
that on every input instance $I\in$ FC-RSA could find a valid coloring $P=P(I)$ that approximates
$opt_{\rm FC}(I)$ (the optimal solution for $I$) so that $|P|\le opt_{\rm FC}(I)^{\alpha} |I|^{\beta}$.
\item There exists 
some $\alpha\ge 1$ and
a polynomial-time quantum algorithm that, on all instances $I\in$ FC-RSA
finds a valid coloring P such that
$|P| < opt_{\rm FC}(I)^{\alpha}$.
\end{enumerate}
\end{thm}

Similar theorems however immediately follow from the PCP theory without further ado. 
Let us take any NP optimization (NPO) problem $\Pi$. As we shall explain, the approximation advantage for quantum methods
(compared to classical)  to 
solve a {\em set of carefully chosen instances} of $\Pi$
is often exactly as large as the inapproximability gap for $\Pi$, i.e. the largest possible. 

Here is  the basis for this:
The PCP theorem as well as a plethora of related 
inapproximability results employ a sequence of Karp reductions \cite{k1972} from the 3SAT problem. Karp's reduction 
besides the instance transformation also features forward and backward witness transformations,
all in polynomial time.  (For readers who are new to these reductions 
we will do a case study of MAX-3SAT, but MAX-CUT, MIN-COLOR, MAX-CLIQUE, etc.\ go exactly in the same way.)

The classically hard instances are nothing but Karp-encoded images 
of Shor-factoring instances where the encoding entails intricate gap-enlarging techniques. These in turn are 
easily solved by a fault tolerant quantum computer by first finding out what the original factoring instance was (this is straightforward), and then
applying Shor's algorithm to it. The encoded instances have no practical value in addition to what their originals have, 
and they are unlikely to occur in nature or elsewhere in any other ways
than as a result of the above encoding process.
(We would love to solve them with a {\em generic noisy algorithm} without going back to the original, but this is not what anyone knows how to do.)

The NP=PCP theorem has
famously put exact and approximate problems on equal footing. With this in mind, claims of solving select hard instances
of an approximation problem with quantum is no more surprising than:
{\it ``A simple quantum advantage for exactly
solving combinatorial optimization problems may be obtained by reducing the integer factoring
problem to 3SAT and leveraging the advantage of Shor's algorithm''} \cite{s2022}. 

We have chosen the case of MAX-3SAT to demonstrate how one easily gets something like Theorem \ref{Pirnayatal}
from known Karp reductions:
\begin{thm}
Assume (as customary), that the factoring problem is not in $P$. 
Then there is subset $S$ of {\rm MAX-3SAT} instances such that:
\begin{enumerate}
\item There is no polynomial time classical algorithm that for all instances $I\in S$ finds an assignment $\vec{a}$ that satisfies more than
$7/8+0.01$ fraction of the clauses of $I$.
\item There is a quantum algorithm that for all instances  $I\in S$ can find an assignment that satisfies at least 99\% of all clauses.
\end{enumerate}
\end{thm}

\begin{proof}
Consider the decision version of the factoring problem with input $(n,k)$: 
\[
(n,k) \rightarrow
\left\{
\begin{array}{ll}
1 & \mbox{if $n$ has a non-trivial factor at most $k$} \\
0 &  {\rm otherwise}  \\
\end{array}
\right.
\]
A certificate for the instance $(n,k)$ is an $f\in\mathbb{N}$ such that $f|n$ and $2\le f\le k$.
This problem is in NP while the 3SAT is NP hard. Thus there is a map $\alpha$ from the factoring problem to the 3SAT problem and 
in addition, a 
certificate transformation $\alpha'$ from certificates of the factoring problem to Boolean assignments for the 3SAT problem, 
such that 
\begin{enumerate}
\item For any satisfiable instance $N=(n,k)$ of the factoring problem and any certificate $f$ of $N$ it holds that
$\alpha(N)[\alpha'(N,f)] = 1$. (In simple words: assignment $\alpha'(N,f)$ satisfies 3SAT instance $\alpha(N)$.)
\item $\alpha$ is poly time computable, and in particular it increases instance sizes only polynomially.
\item For an instance $J$ of 3SAT we can in polynomial time tell if there is an
instance $N$ of the factoring problem such that $\alpha(N) = J$, and if there is such, it computes it.
\item For an instance $N$ of the factoring problem the 3SAT instance $J = \alpha(N)$ is satisfiable if
and only if $N$ is a positive instance.
\end{enumerate}

Further (analogously to $\alpha$ and $\alpha'$, and also in Karp-fashion), due to J. H{\aa}stad \cite{h2001}
we have a ``gap enlarging'' map $\beta$ that maps 3SAT instances to MAX-3SAT instances 
(called {\it instance transformation}), and a transformation $\beta'$ from Boolean assignments (denoted by $\vec{a}$) to Boolean assignments 
(called {\it witness transformation}; the length of $\beta'(J,\vec{a})$ is typically larger than the length of $\vec{a}$) such that:
\begin{enumerate}
\item For any satisfiable instance $J$ of the 3SAT problem and any satisfying assignment $\vec{a}$ of $J$ it holds
that the assignment $\beta'(J,\vec{a})$ satisfies at least 99\% of the clauses of $\alpha(J)$.
\item $\beta$ is poly time computable, and in particular it increases instance sizes only polynomially.
\item For an instance $I$ of MAX-3SAT we can tell in polynomial time if there is an
instance $J$ of 3SAT such that $\beta(J) = I$, and if there is such, it computes it.
\item If $J$ is a negative 3SAT instance (i.e, unsatisfiable), then at most $7/8+0.001$ fraction of the clauses of $I = \beta(J)$ are satisfiable.
\end{enumerate}
\begin{dfn}[The definition of $S$]
We define $S$ required by the theorem as 
\[
S = \{I\in \text{MAX-3SAT}\mid \exists N:\; I = \beta\circ \alpha (N)\}
\]
\end{dfn}
Assume the opposite of item 1 in the theorem. Then we could tell apart
$\ge 99$\% satisfiable instances of $S$ from $\le 7/8+0.001$ fraction satisfiable instances in polynomial time.
Since by applying 
$\beta\circ\alpha$ on any factoring instance $N$ we get an instance in $S$,
which by the properties of $\alpha$ and $\beta$ belongs to the 99\% satisfiable category
if and only if $N$ is a positive instance, by the above 
hypothetised approximation procedure we would be able to tell in polynomial time
if $N$ is a positive instance or not. This stands in contradiction to the assumption of the theorem that
factoring is not in P.

Item 2 of the theorem follows from
the following algorithm to solve any $I\in S$ with a fault tolerant quantum device:
\begin{enumerate}
\item Given $I\in S$ find instance $N$ of the factoring problem with the property that $I = \beta\circ \alpha (N) = \beta(\alpha (N))$. (This is easy.)
\item Using Shor's algorithm find a witness $f$ if $N$ is a positive instance. Note: Shor finds such a witness if exists. For $N=(n,k)$ it just compares 
the smallest factor of $n$ with $k$.
\item Output a random assignment for $I$ if $N$ was a negative instance.
\item If $N$ is a positive instance, output $\beta'(\alpha(N),\alpha'(N,f))$. 
\end{enumerate}
Both for the negative and positive instances
99\% quantum approximation is achieved {\em on $S$}.\end{proof}

\begin{rem}
Although the above example only establishes a factor of $\approx 8/7$ quantum approximation advantage, for the MAX-3SAT problem this is optimal, due to a 
trivial classical algorithm that matches this bound. In a similar fashion, 
{\em super-polynomial} quantum approximation advantage can be achieved for the MAX-CLIQUE and MIN-COLORING problems, as there is a sequence of 
Karp reductions leads to the corresponding near-optimal gap-problems 
(GAP-MAX-CLIQUE$_{n^{1-\epsilon},n^{\epsilon}}$, GAP-MIN-COLOR$_{n^{1-\epsilon},n^{\epsilon}}$) \cite{h1999}
from the decision version of Factoring. {\em However}, with approximation problems that are only unique game (UG) hard \cite{k2002}
the situation is different, as to the knowledge of the author there is no known reduction from the factoring to UGP. MAX-CUT 
falls into this category. To show the possibly optimal $\approx 1/0.878$ quantum approximation advantage for the MAX-CUT seems out of reach,
although smaller constant factor advantage is provable.
\end{rem}

Given our poor understanding of the above ``quantum advantage'' for
NPO problems whose inapproximability is derived from UG rather than 
NP complete problems, it would be interesting to solve the following problem 
(a negative result or putting UG into BQP would be also interesting):

\medskip

{\noindent \bf Problem:}  Give a Karp reduction from the decision version of the factoring problem to UG.

\end{document}